\documentclass[preprint,12pt]{elsarticle}

\usepackage{amssymb}
\usepackage{amsmath}
\usepackage{amsfonts}
\usepackage{amsthm}
\newtheorem{mydef}{Definition}
\newtheorem{thm}{Theorem}
\newtheorem*{them}{Theorem}

\usepackage{calc}
\usepackage{graphicx}
\usepackage{tikz}

\usepackage{pgf}
\usetikzlibrary{fit}
\usetikzlibrary{intersections}
\usetikzlibrary{decorations.markings}
\usetikzlibrary{mindmap}
\usepackage{verbatim}
\usetikzlibrary{positioning,shapes,shadows,arrows}

\usepackage{color}

\tikzstyle{vertex}=[circle, draw, inner sep=0pt, minimum size=6pt]
\newcommand{\vertex}{\node[vertex]}

\usepackage{algorithm}
\usepackage{algpseudocode}
\usepackage{caption}

\usepackage[retainorgcmds]{IEEEtrantools}

\allowdisplaybreaks

\usepackage{url}

\journal{Theoretical Computer Science}

\begin{document}

\begin{frontmatter}
\title{\textbf{Moving in temporal graphs with very sparse random availability of edges}}
\author[1,2]{Paul G. Spirakis}
\author[3]{Eleni Ch. Akrida}
\address[1]{\small Computer Technology Institute \& Press “Diophantus” (CTI), Patras, Greece}
\address[2]{\small Department of Computer Science, University of Liverpool, UK}
\address[3]{\small Department of Mathematics, University of Patras, Greece

Email: \url{spirakis@cti.gr}, \url{akridel@master.math.upatras.gr}}

\begin{abstract}
\footnotesize{
In this work we consider temporal graphs, i.e. graphs, each edge of which is assigned a set of discrete time-labels drawn from a set of integers. The labels of an edge indicate the discrete moments in time at which the edge is available. We also consider temporal paths in a temporal graph, i.e. paths whose edges are assigned a strictly increasing sequence of labels. Furthermore, we assume the uniform case (UNI-CASE), in which every edge of a graph is assigned exactly one time label from a set of integers and the time labels assigned to the edges of the graph are chosen randomly and independently, with the selection following the uniform distribution. We call uniform random temporal graphs the graphs that satisfy the UNI-CASE. We begin by deriving the expected number of temporal paths of a given length in the uniform random temporal clique. We define the term temporal distance of two vertices, which is the arrival time, i.e. the time-label of the last edge, of the temporal path that connects those vertices, which has the smallest arrival time amongst all temporal paths that connect those vertices. We then propose two statistical properties of temporal graphs. One is the maximum expected temporal distance which is, as the term indicates, the maximum of all expected temporal distances in the graph. The other one is the temporal diameter which, loosely speaking, is the expectation of the maximum temporal distance in the graph. Since uniform random temporal graphs, except for the clique, have at least a pair of vertices whose temporal distance is infinity, we assume the existence of a \textit{slow} way to go directly from any vertex to any other vertex in order for the above measures to have a finite value. We derive the maximum expected temporal distance of a uniform random temporal star graph as well as an O($\sqrt{n} \log^2{n}$) upper bound, and a greedy algorithm which computes in polynomial time the path that achieves it, on both the maximum expected temporal distance and the temporal diameter of the \textit{normalized} version of the uniform random temporal clique, in which the largest time-label available equals the number of vertices. Finally, we provide an algorithm that solves an optimization problem on a specific type of temporal (multi)graphs of two vertices.}
\end{abstract}

\end{frontmatter}

\section{Introduction}

A temporal graph (or otherwise called temporal network) is, loosely speaking, a graph that changes with time. This concept incorporates a variety of both modern and traditional networks such as information and communication networks, social networks, transportation networks, and several physical systems. The presence of dynamicity in modern communication networks, i.e. in mobile ad hoc, sensor, peer-to-peer, and delay-tolerant networks, is often very strong. We can also find that kind of dynamicity in social networks, where the topology usually represents the social connections between a group of individuals. Those connections change as the social relationships between the individuals or even the individuals themselves change. Temporal graphs can also be associated with transportation networks. In a transportation network, there is usually some fixed network of routes and a set of transportation units moving over these routes. In such networks, the dynamicity refers to the change of positions of the transportation units in the network as time passes. Concerning physical systems, dynymicity may be present in systems of interacting particles.

In this work, embarking from the foundational work of Kempe et al. \cite{kempe}, we consider the time to be discrete, that is, we consider networks in which changes can only occur at discrete moments in time, e.g. days or hours. This choice not only gives to the resulting models a purely combinatorial flavor but also naturally abstracts many real systems. In particular, we consider those networks that can be described via an underlying graph $G$ and a labeling $L$ assigning a set of discrete labels to each edge of $G$. This is a generalization of the single-label-per-edge model used in \cite{kempe}, as we allow many time-labels to appear on an edge, although in this work we mainly focus on single-labeled temporal graphs. These labels are drawn from the natural numbers and indicate the discrete moments in time at which the corresponding connection is available, i.e. the corresponding edge exists in the graph. For example, in a communication network, the availability of a connection at some time $t$ may indicate that a communication protocol is allowed to transmit a data packet over that connection at time $t$. A temporal path (or journey) in a temporal graph is a path, on the edges of which we can find strictly ascending time labels. The number of edges on the latter is called length of the temporal path. This, for a communication network, would mean that it is possible to transmit a data packet along the network nodes that belong to such a path from the first node in order to the last one, as time progresses. The time label on the last edge of a temporal path is called its arrival time and, in the above example of a connection network, it would indicate the time at which the transmitted data packet would arrive at the last node of the path.

In this work, we initiate the study of temporal graphs from a probabilistic and statistical viewpoint. In particular, we consider the case in which every edge of a graph is assigned exactly one time label from a set $L_0 = \{1, 2, \ldots, a\}$ of integers. The time labels assigned to the edges of the graph are chosen \textit{randomly} and \textit{independently} from one another from the set $L_0$ and the probability that an edge is assigned a time label $i \in L_0$ is equal to $\frac{1}{a}$, for every $i \in L_0$. We use the term \emph{UNI-CASE} for the  above described case and for any graph that satisfies UNI-CASE's properties we use the term \emph{Uniform Random Temporal Graph}. We focus on examining three statistical properties of such graphs. The first one, called \emph{expected number of temporal paths of a given length}, is the number of temporal paths, of a given length, that we expect to have in a graph, given that every edge is assigned a label satisfying UNI-CASE. The second one, called \emph{the Maximum Expected Temporal Distance}, is the maximum of all temporal distances in the graph. By temporal distance of two vertices we denote the arrival time of the temporal path that connects those vertices, which has the smallest arrival time amongst all temporal paths that connect those vertices. The last property that we examine is called \emph{the Temporal Diameter} of a uniform random temporal graph. Loosely speaking, it is the expected value of the maximum temporal distance in the graph, which of course is in correspondence with the diameter of a graph, as we know it up to now.

The motivation of the definitions we initiate and the work we carry out here comes from the natural question on how fast we can visit a particular destination, i.e. arrive at a particular network node, starting from a given point of origin, i.e. another network node, when the connection between a pair of nodes only exists at one moment in time.

\subsection{Related work}

\noindent \textbf{Labeled Graphs.} Labeled graphs are becoming an increasingly useful family of Mathematical Models for a broad range of applications both in Computer Science and in Mathematics, e.g. in Graph Coloring\cite{molloy}. In our work, labels correspond to time moments of availability and the properties of labeled graphs that we study are naturally \emph{temporal properties}. However, we can note that any property of a graph that is assigned labels from a discrete set of labels can correspond to some temporal property. Take for example a proper edge-coloring in a graph, i.e. a coloring of the graph's edges in which no two adjacent edges have the same color. This corresponds to a temporal graph in which no two adjacent edges have the same time label, that is no two adjacent edges exist at the same time.

\noindent \textbf{Single-labeled and multi-labeled Temporal Graphs.} The model of temporal graphs that we consider in this work has a direct relation with the single-labeled model studied in \cite{kempe} as well as the multi-labeled model studied in \cite{spirakis}. The main results of \cite{kempe} and \cite{spirakis} have to do mainly with connectivity properties and/or cost minimization parameters for temporal network design. In this work we study temporal graphs from a statistical view and mainly focus on how fast we expect to arrive at a target vertex in a temporal graph. In \cite{kempe}, a temporal path is considered to be a path with non-decreasing labels on its edges. In this work, we follow the assumption of \cite{spirakis} and consider a temporal path to be a path with strictly increasing labels. This choice is also motivated by recent work on dynamic communication systems, in which if it takes one time unit for the transmition of a data packet over a link, then a packet can only be transmitted over paths with strictly increasing labels.

\noindent \textbf{Continuous Availabilities (Intervals).} Some authors have assumed the availability of an edge for a whole time-interval [$t_1,t_2$] or multiple such time-intervals. Although this is a clearly natural assumption, in this work we focus on the availability of edges at discrete moments and we design and develop techniques which are quite different from those needed in the continuous case.

\subsection{Roadmap and contribution}
In Section \ref{sec:pre}, we formally define the model of temporal graphs under consideration and provide all further necessary basic definitions. In Section \ref{sec:exp}, we make some general remarks on the expected number of temporal paths in any graph and proceed to the study of the expected number of temporal paths of a given length in the uniform random temporal clique of $n$ vertices, $K_n$. For this matter, we distinguish two cases. In Section \ref{sec:exp1}, we study the first case, where we set the largest label available for assignment to be $a=n-1$ and we search for the expected number of temporal paths of length $k=n-1$. In Section \ref{sec:exp2}, we study the second case, where we loosen the parameters $a$ and $k$ and we look at the expected number of temporal paths of length $k<a$, when the largest label available for assignment is $a=n-1$. In Section \ref{sec:md}, we formally define the maximum expected temporal distance of a uniform random temporal graph and we make some preliminary notations. In Section \ref{sec:md1}, we look at some known graphs' maximum expected temporal distance. In particular, in Section \ref{sec:md11}, we study the case of the uniform random temporal star graph and we provide its exact maximum expected temporal distance. In Section \ref{sec:md12}, we study the case of the uniform random temporal clique, focusing on its normalized version, where the largest label, $a$, available for assignment is equal to the number of vertices, $n$. We also give a simple \textit{(greedy)} algorithm which can, with high probability, find a temporal path with small expected arrival time from a given source to a given target vertex in the normalized uniform random temporal clique. In Section \ref{sec:td}, we formally define the temporal diameter of a uniform random temporal graph and provide an inequality relation between the latter and the maximum expected temporal distance as well as the relevant proof. Furthermore, we provide an upper bound for both the temporal diameter and the maximum expected temporal distance of the nomalized uniform random temporal clique. In Section \ref{sec:bridge}, we study an optimization problem on a specific type of temporal (multi)graphs of two vertices. We prove that the problem can by polynomially solved and provide an algorithm that gives the solution, along with the proof of its correctness. Finally, in Section \ref{sec:concl} we conclude and give further research horizons opened through our work.

\section{Preliminaries}\label{sec:pre}

\begin{mydef}
A temporal graph is an ordered triplet $G=\{V,E,L\}$, where:
\begin{itemize}
\item $V$ stands for a nonempty finite set (called set of vertices)
\item $E$ stands for a set of m elements, each of which is a 2-element subset of V (called set of edges), and
\item $L= \{L_e, \forall e \in E\} = \{L_{e_1}, L_{e_2}, \ldots, L_{e_m}\}$, is a set of m elements, $L_{e_i},~1\leq i \leq m$, each of which is a set of positive integers  mapped to the edge $e_i \in E$ (called assignment of time labels or simply assignment)
\end{itemize}
\end{mydef}

We also denote the temporal graph $G=\{V,E,L\}$ by $G'(L)$ or $(G',L)$, where $G' = \{V,E\}$ is the graph, on the edges of which we assign the time labels, and $L= \{L_e, ~ e \in E(G')\}$ is the assignment.

The values assigned to each edge of the graph are called time labels of the edge and indicate the times at which we can cross it (from one end to the other).

\subsection{Further Definitions}\label{sec:def}
We can now talk about temporal edges (or time edges) that are considered to be triplets $(u, v, l)$, where $u, v$ are the ends of an edge in the temporal graph and $l \in L_{ \{u, v\} }$ is a
time label of this edge. That is, if an edge $e = \{u, v\}$ has more
than one time labels, e.g. has a set of three time labels, $L_e = \{l_1, l_2, l_3\}$, then this edge has three corresponding time edges, $(u, v, l_1),~ (u, v, l_2)$ and $(u, v, l_3)$.

\begin{mydef}
A journey $j$ from a vertex $u$ to a vertex $v$ ($(u, v)$-
journey) is a sequence of time edges $(u, u_1, l_1), ~(u_1, u_2, l_2), \ldots , ~(u_{k-1}, v, l_k)$, such that $l_i < l_{i +1}$, for each $1 \leq i \leq k - 1$.\\
We call the last time label of journey $j$, $l_k$, {\emph arrival time} of the journey.
\end{mydef}

\begin{mydef}
A ($u,v$)-journey $j$ in a temporal graph is called {\emph foremost journey} if its arrival time is the minimum arrival time of all ($u,v$)-journeys' arrival times, under the labels assigned on the graph's edges.
\end{mydef}

Now, consider any temporal graph $G=\{V,E,L\}$.
Let every edge receive exactly one time label, chosen randomly, independently of one another from a set $ L_0 $ = \{$ 1,2, \ldots, a $\}, where $ a \in \mathbb{N} $, with the probability of an edge label to be $ i, ~ \forall i \in L_0 $, equal to $ \frac{1}{a} $. (\textbf {UNI-CASE)}

\begin{mydef}
A temporal graph that satisfies UNI-CASE is called \textit{Uniform Random Temporal Graph (U-RTG)}.
\end{mydef}

In the special case, where the largest label, $a$, that can be assigned to the edges of a graph is equal to the number of its vertices, the graph is called \textit{Normalized Uniform Random Temporal Graph (Normalized U-RTG)}.

\textbf {Note.} There could be prospective study of cases in which each edge of a graph may receive several time labels, selected randomly and independently of one another from the set $ L_0 $ = \{$ 1 , 2, \ldots, a $\}, where $ a \in \mathbb {N} $, with the selection following a distribution F. (\textbf {F-CASE)}\\ In such cases, the graphs under consideration would be called \textit{F-Random Temporal Graphs (F-RTG)} respectively. \label{sec:eisag}

In the following sections, we will look for the expected number of journeys of length k in some well-known graphs that satisfy UNI-CASE. For the sake of brevity, we often call such journeys \textit {``k edges temporal paths''}. \ \
We also study the Expected (or Temporal) Diameter and the Maximum Temporal Distance of a graph, as defined in the following paragraphs.

\section{Expected number of temporal paths}\label{sec:exp}

In this section we will search for the expected number of $k$ edges temporal paths in a clique of n vertices, $ K_n $, that satisfies UNI-CASE.

It is obvious that for there to exist a temporal path of  length k in \textbf{any} graph, the number of edges, k, has to be at most equal to the maximum label of the set $ L_0 $, $ a $,  that can be assigned to the various edges. Otherwise, it is impossible for a $k$ edges temporal path to exist (see Figure \ref{fig:temp-stat1}).

\begin{figure}[htbp]
\begin{center}

\[\begin{tikzpicture}[thick,scale=0.95]
	\coordinate [label=right:{$L_0=\{1,2,3 (=a)\}$}] (P) at (0,-1);
	\coordinate [label=right:{$k=4$}] (P') at (0,-2);
	\vertex (1) at (0,0) [label=below:$$] {};
	\vertex (2) at (2,0.5) [label=left:$$] {};
	\vertex (3) at (4,0.2) [label=left:$$] {};
	\vertex (4) at (6,0) [label=above:$$] {};
	\vertex (5) at (8,0.4) [label=above:$$] {};
	\path
		(1) edge node[above]{$1$} (2)
		(2) edge node[above]{$2$} (3)
		(3) edge node[above]{$3$} (4)
		(4) edge node[above]{\textbf{\textcolor{red}{;}}} (5)
	;	
\end{tikzpicture}\]
\end{center}

\rule{35em}{0.5pt}
\caption{There is no temporal path, when $k>a$}
\label{fig:temp-stat1}
\end{figure}
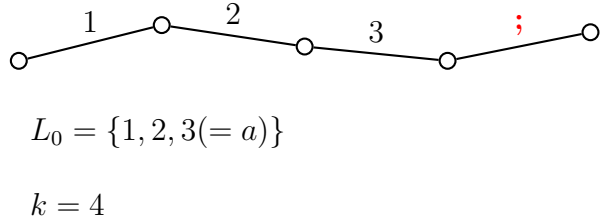

\subsection{Special case: $G=K_n,~k=n-1,~a=n-1$}\label{sec:exp1}

Initially, we focus our interest in the case of the clique (complete graph) of n vertices, $ K_n $, that satisfies UNI-CASE with $a = n-1$ (i.e. with $ L_0 = \{1,2, \ldots, n-1 \} $), in which we seek the expected number of $n-1$ edges temporal paths.

Obviously, there can only be one assignment of labels of $ L_0 $ on the $ k = n-1 $ edges of any path starting from a random initial vertice $ v_0 \in V (K_n) $ in the clique $ K_n $ such, that we can find a journey on the edges of this path. This assignment gives label 1 on the $1^{st}$ edge, label 2 on the $2^{nd}$ edge, $\ldots$ , label $n-1$ on the $(n-1)^{th}$ edge.

Each edge can receive exactly one label from a set of $n-1$ labels. Therefore, the total number of assignments that can be made on these $ n-1 $ edges is:
\[ \# assignments = (n-1)^{n-1} \]

Consequently, given a path of $ n-1 $ edges starting from $ v_0 $, the probability for there to exist the corresponding temporal path (i.e. the one arising on the simple path after the assignment of the time labels) is:
\[ P(temporal\_ path\_ of\_length \_n-1 \_starting\_ from \_v_0)= \frac{1}{(n-1)^{n-1}} \]

The number of paths of length $ n-1 $, starting from $ v_0 $ in the clique $ K_n $  is equal to the number of permutations of the $ n-1 $ vertices remaining (i.e. except the start $ v_0 $) to construct such a path. That is, the number of paths of length $ n-1 $ that start from $ v_0 $ in the clique $ K_n $  is:
\[ (n-1)! \]

Therefore, since the clique $ K_n $ has $ n $ vertices, and due to the linearity of expectation, the expected number of temporal paths of length $ k = n-1 $ in the clique $K_n$ is:
\[ E(\# temporal \_paths\_of\_length\_n-1) = n \cdot (n-1)! \cdot \frac{1}{(n-1)^{n-1}} = \frac{n!}{(n-1)^{n-1}}\]

\paragraph{Comments} Let us observe that when $ n $ is too large ($n\rightarrow + \infty$), then, by Stirling's formula, we result in the following:
\begin{IEEEeqnarray*}{lCl}
E(\# temporal \_paths\_of\_length\_n-1) & = & \frac{\sqrt{2\pi n} \Big(\frac{n}{e}\Big)^n}{(n-1)^{n-1}}	
\\
 & = &	\frac{\sqrt{2\pi n} n^n}{e^n (n-1)^{n-1}} \xrightarrow[n\to+\infty]{} 0
\end{IEEEeqnarray*}
Of course, this is more or less obvious when we consider the fact that it is difficult to find $n-1$ edges temporal paths in the clique of $n$ vertices when $n$ is too large. This is because in order to have a temporal path of such length, the (so many) time labels should be assigned on the edges so that they maintain the desired strictly increasing sequence, something that is increasingly less likely to happen as $ n $ increases.

\subsection{Special case: $G=K_n,~k<a,~a\geq n$}\label{sec:exp2}
Now let's see what happens in the case of the clique $K_n$, that satisfies UNI-CASE, when we look at the expected number of temporal paths of length $k<a$ and the maximum label that can be assigned to any edge of the clique is $a \geq n$.

Starting from a vertex $v_0 \in V(K_n)$ and along the path of k edges, we can construct, as explained in Figure \ref{fig:temp-stat2}, a number of assignments equal to:
\[ \# assignments = a^k \]

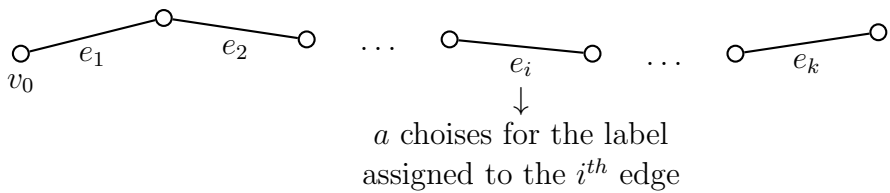
\begin{figure}[htbp]
\begin{center}

\[\begin{tikzpicture}[thick,scale=0.95]
	\coordinate [label=below:{$\ldots$}] (P) at (5,0.25);
	\coordinate [label=below:{$\ldots$}] (P') at (9,0.05);
	\coordinate [label=below:{$\downarrow$}] (P'') at (7,-0.3);
	\coordinate [label=below:{$a$ choises for the label}] (P'') at (7,-0.8);
	\coordinate [label=below:{assigned to the $i^{th}$ edge}] (P'') at (7,-1.3);
	\vertex (1) at (0,0) [label=below:$v_0$] {};
	\vertex (2) at (2,0.5) [label=left:$$] {};
	\vertex (3) at (4,0.2) [label=left:$$] {};
	\vertex (4) at (6,0.2) [label=above:$$] {};
	\vertex (5) at (8,0) [label=above:$$] {};
	\vertex (6) at (10,0) [label=above:$$] {};
	\vertex (7) at (12,0.3) [label=above:$$] {};
	\path
		(1) edge node[below]{$e_1$} (2)
		(2) edge node[below]{$e_2$} (3)
		(4) edge node[below]{$e_i$} (5)
		(6) edge node[below]{$e_k$} (7)
	;	
\end{tikzpicture}\]
\end{center}

\rule{35em}{0.5pt}
\caption{Number of assignments on a path of length $k$, when $k<a$}
\label{fig:temp-stat2}
\end{figure}

The number of assignments that can be made on the $ k $ edges, where the time labels assigned are distinct (different from each other) is:
\[ \# distinct\_ time\_ labels\_assignments = a \cdot (a-1) \cdot \ldots \cdot (a-k+1) = \frac{a!}{(a-k)!} \]

We will now calculate the number of paths of length $k$ that can be starting from $v_0 \in V(K_n)$. We have $n-1$ options for how to select $v_1$, the vertex following $v_0$ on the path, $n-2$ options for how to select $v_2$, the vertex following $v_1$ on the path, etc., and finally $n-k$ options for how to select $v_k$, the last vertex on the path.

Therefore, the number of paths of length $ k $ that can be starting from $ v_0 \in V (K_n) $ is:
\[ \#paths \_of\_length\_k\_starting\_from\_v_0 = (n-1)\cdot (n-2) \cdot \ldots \cdot (n-k) =\frac{(n-1)!}{(n-k-1)!} \]

We call $A$ the event that ``we have the \textit{right} labels'' assignment on the $k$ edges of any path of length $k$ starting from $ v_0 $''. \\
That is, if $l_1, l_2, \ldots, l_k$ are the time labels assigned to the $1^{st}$, the $ 2^{nd} $, $ \ldots $, the $ k^{th} $ edge of the path, respectively, with $ l_i \in L_0 = \{1,2, \ldots, a \}, ~ \forall i = 1,2, \ldots, k $, $A$ is the event that:
\[l_1 < l_2 < \ldots < l_k\]

We call $\phi$ the probability that $A$ occurs. That is:
\[ \phi = P(A)= P(l_1 < l_2 < \ldots < l_k) \]

Let us note that the number of assignments of $k$ labels, $l_{a_i},~ i=1, \ldots, k$, such that \[ l_{a_1} < l_{a_2} < \ldots < l_{a_k} \]
is $k!$ and each one has a probability equal to $P(A)$ to happen.\\
Therefore, if we consider $B$ to be the event that ``\textit{at least} two of the labels assigned on the $k$ edges of the path are equal'', then the following applies:
\[ k! \cdot P(A) + P(B) =1 \Leftrightarrow \]
\begin{equation}\label{eq:1}
k! \cdot \phi + 1 - P(\rceil{B}) =1
\end{equation}

The probability that the event $\rceil{B}$ occurs, that is there are no two equal labels assigned on the $k$ edges of the path, is:
\[ P(\rceil{B}) = \frac{\#distinct\_ time\_labels\_ assignments}{\# assignments} =\]
 \[ = \frac{\frac{a!}{(a-k)!}}{a^k} \]
 \[ = \frac{a!}{a^k \cdot (a-k)!} \]

Consequently, the relation \eqref{eq:1} becomes:
\[ k! \cdot \phi + 1 - \frac{a!}{a^k \cdot (a-k)!} =1 \Leftrightarrow \]
\[ \Leftrightarrow \phi = \frac{a!}{k! \cdot a^k \cdot (a-k)!}\]

Let us recall that $\phi$ is the probability to have a \textit{proper} assignment on the $ k $ edges of any path of length $k$ starting from any vertice $ v_0 $ of the clique $ K_n $.\\
Also, recall that the number of paths of length $k$ that can be starting from any vertice $v_0$ of the clique $K_n$ is $\frac{(n-1)!}{(n-k-1)!}$.\\
Therefore, the expected number of paths of length $k$ that start from a random vertex $v_0$ and on which there are labels assigned so that there exists a temporal path on them, is:
\[ E(\#temporal\_paths\_of\_length\_k\_starting\_from\_v_0) = \frac{(n-1)!}{(n-k-1)!} \cdot \phi \]

Eventually, since the clique $K_n$ has a number of $n$ vertices, the expected number of paths of length $k$, on which labels are assigned in a way that there exists a temporal path on them, is:
\[ E(\#temporal\_paths\_of\_length\_k) = n \cdot \frac{(n-1)!}{(n-k-1)!} \cdot \phi \]
\[ = \frac{n \cdot (n-1)!}{(n-k-1)!} \cdot \frac{a!}{k! \cdot a^k \cdot (a-k)!} \]
\[ = \frac{n! \cdot a!}{(n-k-1)! \cdot k! \cdot a^k \cdot (a-k)!} \]

\paragraph{Comments} Let us observe that the probability $\phi$ is:
\[ \phi = \frac{1}{k!} \cdot \frac{\overbrace{a (a-1) \ldots (a-k+1)}^{\text{k factors}}}{\underbrace{a \cdot \ldots \cdot a}_\text{k factors}}\]
and so, if $a$ is very large in comparison with $k$, then we have $\phi  \approx \frac{1}{k!}$.\\
Hence, if $a$ is far larger than $k$, then the expected number of temporal paths of length $k$ in the clique $K_n$, is:
\[ E(\#temporal\_paths\_of\_length\_k) \approx \frac{n!}{k! (n-k-1)!} = \frac{n \cdot (n-1) \cdot \ldots \cdot (n-k)}{k!} \]

\section{The Maximum Expected Temporal Distance}\label{sec:md}

In this section, we will define and study a new concept, that of \textit {the maximum expected temporal distance} of a U-RTG.\\
Henceforth, we make the following assumption. For every pair of vertices in any U-RTG, there exists a \textit{\textbf{slow}} journey that connects them, whose arrival time is a fixed, for each graph, number $n' \in \mathbb{N},$ where $n'$ is greater than the expected value of any edge's label, $l$. That is $n' \geq E(l)$.

\begin{mydef}
Consider an instance $G(L)$ of a U-RTG. Given two vertices $s,t \in V \big( G(L) \big)$, we define:
\begin{itemize}
\item $\delta ' (s,t)=a(j),$\label{s10} where $j$ is a foremost $(s,t)-$journey, to be called \textbf{\textit{distributional temporal distance}} from source vertex $s$ to target vertex $t$ under the assignment $L$. If there exists no $(s,t)-$journey in G, then $\delta ' (s,t) \rightarrow \infty $

\item $ \delta(s,t) = min \{ \delta ' (s,t) , n' \} $ to be called \textbf{\textit{temporal distance}} from source vertex $s$ to target vertex $t$ under the assignment $L$, and

\item $MD= max_{s,t \in V(G)} E\big(  \delta(s,t) \big)$\label{s11} to be called \textbf{\textit{Maximum Expected Temporal Distance}} of $G$
\end{itemize}
\end{mydef}

\noindent \textit{Remark.} If the $U-RTG$ is a path itself, then its maximum expected temporal distance is obviously $n'$. \\

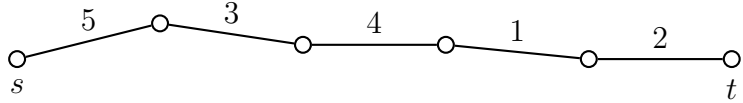
\begin{figure}[htbp]
\begin{center}

\[\begin{tikzpicture}[thick,scale=0.95]
	\vertex (1) at (0,0) [label=below:$s$] {};
	\vertex (2) at (2,0.5) [label=left:$$] {};
	\vertex (3) at (4,0.2) [label=left:$$] {};
	\vertex (4) at (6,0.2) [label=above:$$] {};
	\vertex (5) at (8,0) [label=above:$$] {};
	\vertex (6) at (10,0) [label=below:$t$] {};
	\path
		(1) edge node[above]{$5$} (2)
		(2) edge node[above]{$3$} (3)
		(3) edge node[above]{$4$} (4)
		(4) edge node[above]{$1$} (5)
		(5) edge node[above]{$2$} (6)
	;	
\end{tikzpicture}\]
\end{center}

\rule{35em}{0.5pt}
\caption{MD of a $U-RTG$, which is a path itself, equals $n'$.}
\label{fig:td1}
\end{figure}

This can be easily understood if we consider that for any two vertices $u$ and $v$ in the path, if there exists a ($u,v$)-journey, then the time labels assigned to its edges form a strictly increasing sequence and thus there is no ($v,u$)-journey in it, apart from the \textit{slow} journey which we assume that exists. Therefore, $\delta ' (v,u) \rightarrow + \infty$ and $ \delta (v,u) = min \{\delta ' (v,u), n'\} = n'$. (see Figure \ref{fig:path11}).

\begin{figure}[htbp]
\begin{center}

\[\begin{tikzpicture}[thick,scale=0.95]
	\vertex (1) at (0,0) [label=below:$$] {};
	\vertex (2) at (2,0.5) [label=left:$$] {};
	\vertex (3) at (4,0.2) [label=above:$u$] {};
	\vertex (4) at (6,0.55) [label=above:$$] {};
	\vertex (5) at (8,0) [label=above:$v$] {};
	\vertex (6) at (10,0) [label=above:$$] {};
	\vertex (7) at (12,0.5) [label=above:$$] {};

	\node[anchor=east] at (2.65,-2.5) (a) {};
	\node[anchor=west] at (9,-2.5) (b) {};
	\node[anchor=east] at (3.3,-1.4) (a') {};
	\node[anchor=west] at (8.2,-1.4) (b') {};

	\vertex (8) at (2.65,-1.9) [label=above:$u$] {};
	\vertex (9) at (6,-1.25) [label=above:$$] {};
	\vertex (10) at (9,-2.1) [label=above:$v$] {};
	\path
		(8) edge node[sloped, above]{\textcolor{blue!70}{$3$}} (9)
		(9) edge node[sloped, above]{\textcolor{blue!70}{$4$}} (10)
		(1) edge node[above]{$$} (2)
		(2) edge  [line width=1pt,black!0.1]  node[sloped, below, black]{$\ldots$}  (3)
		(3) edge node[sloped, above]{\textcolor{blue!70}{$3$}} (4)
		(4) edge node[sloped, above]{\textcolor{blue!70}{$4$}} (5)
		(5) edge  [line width=1pt,black!0.1]  node[sloped, below, black]{$\ldots$} (6)
		(6) edge node[above]{$$} (7)

		(3) edge [line width=1pt,dotted, blue!60] node[above]{$$} (8)
		(5) edge [line width=1pt, dotted, blue!60] node[above]{$$} (10)
		(a) edge[->, bend left=20, blue!70] node [below]{\textcolor{blue!70}{$\delta ' (u,v) = 4$}} (b)
		(b') edge[->, bend right=25, red] node [above]{\textcolor{red}{$\delta ' (v,u) \rightarrow \infty $}} (a');
	;
\end{tikzpicture}\]
\end{center}
\rule{35em}{0.5pt}
\caption{Example of temporal distance, from source vertex to target vertex, equal to $ n'$.}
\label{fig:path11}
\end{figure}
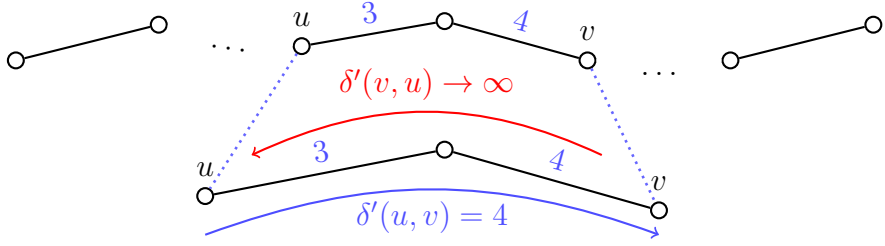

\subsection{Known graphs' maximum expected temporal distance}\label{sec:md1}
Next, we study the maximum expected temporal distance of two known graphs, the star graph of $n$ vertices, which we denote by $G_{star}$ (see Figure \ref{fig:td2}) and the clique of $n$ vertices, $K_n$ (see Figure \ref{fig:td3}).\\

\subsubsection{Case: $G=G_{star}$}\label{sec:md11}

It is easy to understand that, even if the temporal star graph does not satisfy UNI-CASE, but satisfies any F-CASE, as defined in Section \ref{sec:eisag}, it is:
\[max_{s,t \in V(G_{star})} E_F \big( \delta (s,t) \big) \geq 2,~ \text{for any distribution }F \]

We will calculate the exact maximum expected temporal distance, $MD$, of a uniform random temporal star graph. It is:
\begin{IEEEeqnarray}{rCl}\label{eq:2}
MD(G_{star}) & = & max_{s,t\in V(G_{star})}E \big(\delta (s,t) \big) \nonumber\\ \quad
& = & E\big(\delta (s,t) \big) \text{, for any two vertices } s,t \in V(G_{star})
\nonumber\\ \quad
& = &  E(l_2|~l_2 > l_1) \cdot P(l_2 > l_1) +n' \cdot P(l_2 \leq l_1)
\end{IEEEeqnarray}

\begin{figure}[htbp]
\begin{center}

\[\begin{tikzpicture}[thick,scale=0.95]
	\vertex (1) at (0,0) [label=below:$$] {};
	\vertex (2) at (2,0) [label=left:$$] {};
	\vertex (3) at (1,1.5) [label=above:$t$] {};
	\vertex (4) at (1,-1.5) [label=above:$$] {};
	\vertex (5) at (-1,-1.5) [label=above:$$] {};
	\vertex (6) at (-2,0) [label=left:$s$] {};
	\path
		(1) edge node[above]{$$} (2)
		(3) edge [line width=1pt,black!0.1]  node[sloped, above, black]{$\ldots$} (6)
		(1) edge node[right]{$l_2$} (3)
		(1) edge node[above]{$$} (4)
		(1) edge node[above]{$$} (5)
		(1) edge node[below]{$l_1$} (6)
	;
	
\end{tikzpicture}\]
\end{center}

\rule{35em}{0.5pt}
\caption{A star graph}
\label{fig:td2}
\end{figure}
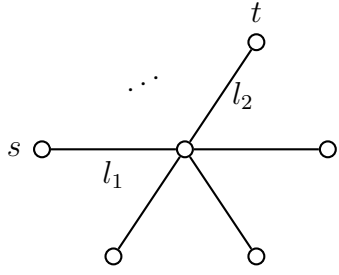

We calculate the expected value of label $l_2$, given that $l_2 > l_1$, that is $E(l_2|~l_2 > l_1)$:
\begin{IEEEeqnarray*}{rCl}
E(l_2|~l_2 > l_1)& = & \sum_{i=1}^a E(l_2|~l_2>i) \cdot P(l_1=i)
\\
& = & \sum_{i=1}^a \Big( \sum_{i'=i}^{a} \big( P(l_2 = i' +1) \cdot (i'+1) \big) \Big) \cdot P(l_1 = i)
\\
& = & \sum_{i=1}^a \Big( \sum_{i'=i}^{a} ( i' +1) \cdot  \frac{1}{a} \Big) \cdot \frac{1}{a}
\\
& = & \frac{1}{a^2} \cdot \sum_{i=1}^{a} \sum_{i'=i}^a (i'+1)
\\
& = & \frac{1}{a^2} \cdot \Big( \sum_{i'=1}^a (i'+1) + \sum_{i'=2}^a (i'+1) + \ldots  + \sum_{i'=a}^a (i'+1)   \Big) 
\\
& = & \frac{1}{a^2} \cdot \Big(  \big( 2+3+ \ldots + (a+1) \big) + \big( 3+4 + \ldots + (a+1) \big) + \ldots + \big( (a+1) \big) \Big)
\\
& = & \frac{1}{a^2} \cdot \Big( 1\cdot 2 + 2 \cdot 3 + 3 \cdot 4 + 4 \cdot 5 + \ldots + a \cdot (a+1) \Big)
\\
& = & \frac{1}{a^2} \cdot \sum_{i=1}^a \Big( i \cdot (i+1) \Big)
\\
& = & \frac{1}{a^2} \cdot \sum_{i=1}^a \Big( i^2 +i \Big)
\\
& = & \frac{1}{a^2} \cdot \sum_{i=1}^a  i^2 + \sum_{i=1}^a i 
 \\
& = & \frac{1}{a^2} \cdot \Big(		\frac{a\cdot (a+1) \cdot (2a +1)}{6}	+		 \frac{a \cdot (a+1)}{2}		\Big)
\\
& = & \frac{1}{a^2} \cdot 	\frac{a\cdot (a+1) \cdot (2a +1)	+	3 \cdot a\cdot (a+1)	}{6}
\\
& = & \frac{a\cdot (a+1) \cdot (2a +4)}{6 \cdot a^2}
\end{IEEEeqnarray*}

Therefore, relation \eqref{eq:2} becomes:
\begin{IEEEeqnarray}{rCl}\label{eq:3}
MD(G_{star}) & = &  \frac{(a+1)(a+2)}{3a}  \cdot P(l_2 > l_1) +n' \cdot P(l_2 \leq l_1)
\end{IEEEeqnarray}

It holds that:
\begin{IEEEeqnarray*}{rCl}
P(l_2 \leq l_1) & = &  \sum_{i=1}^{a} P(l_2 \leq i) \cdot P(l_1 = i)
\\
 & = & \sum_{i=1}^{a} \frac{i}{a} \cdot \frac{1}{a}
\\
 & = & \frac{1}{a^2} \sum_{i=1}^{a} i
\\
 &= & \frac{a+1}{2a}
\end{IEEEeqnarray*}

Therefore, it is:
\begin{IEEEeqnarray*}{rCl}
P(l_2 > l_1) & = &  1- P(l_2 \leq l_1)
\\
 & = &  \frac{a-1}{2a}
\end{IEEEeqnarray*}

Relation \eqref{eq:3} now becomes:
\begin{IEEEeqnarray*}{rCl}
MD(G_{star}) & = &  \frac{(a+1)(a+2)}{3a}  \cdot \frac{a-1}{2a} +n' \cdot \frac{a+1}{2a}
\end{IEEEeqnarray*}

Eventually, the star graph's maximum temporal distance is:
\[  MD(G_{star}) =  \frac{(a-1)(a+1)(a+2)}{6a^2} +n' \cdot \frac{a+1}{2a}\]

\subsubsection{Case: $G=K_n$}\label{sec:md12}

We will now study extensively the clique's case. First, let us observe that $\delta ' (s,t) \leq a$, and therefore $\delta  (s,t) \leq a$, for any two vertices $s,t$ in a clique. Hence:
\[ MD(K_n) = max_{s,t\in V(K_n)}E \big(\delta (s,t) \big)  \leq a \]

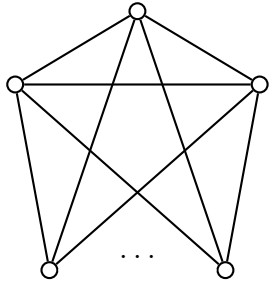
\begin{figure}[htbp]
\begin{center}

\[\begin{tikzpicture}[thick,scale=0.65]
	\vertex (1) at (-0.5,0) [label=below:$$] {};
	\vertex (2) at (2,1.5) [label=left:$$] {};
	\vertex (3) at (4.5,0) [label=left:$$] {};
	\vertex (4) at (3.8,-3.8) [label=above:$$] {};
	\vertex (5) at (0.2,-3.8) [label=above:$$] {};
	\path
		(1) edge node[above]{$$} (2)
		(1) edge node[above]{$$} (3)
		(1) edge node[above]{$$} (4)
		(1) edge node[above]{$$} (5)
		(2) edge node[above]{$$} (3)
		(2) edge node[above]{$$} (4)
		(2) edge node[above]{$$} (5)
		(3) edge node[above]{$$} (4)
		(3) edge node[above]{$$} (5)
		(4) edge [line width=1pt,black!0.1]  node[sloped, above, black]{$\ldots$} (5)

	;
\end{tikzpicture}\]
\end{center}

\rule{35em}{0.5pt}
\caption{A clique}
\label{fig:td3}
\end{figure}

\paragraph{Normalized uniform random temporal clique}
Let $G=K_n$ be a clique of $n$ vertices and let us consider its normalized U-version. That is, every edge $e \in E(K_n)$ is given a single availability label and those labels are chosen randomly and independently from one another from the set $L_0$=\{$1,2, \ldots, n$\}, with the probability that an edge's label equals $i$ being equal to $\frac{1}{n}$, $\forall i \in L_0$.

For any two vertices $s,t$ in the clique, we have:\[ E\big( l (e= \{s,t\}) \big) = \frac{n}{2} \]

In the specific case of the normalized uniform random temporal clique of $n$ vertices, there is actually no need for us to assume any \textit{slow} journey to connect any pair of vertices since we already have such a journey, with arrival time equal to $E\big( l (e= \{s,t\}) \big) = \frac{n}{2}$. But, for the sake of consistency, we can set the fixed number $n'$ to be equal to $\frac{n}{2}$.

It holds that:\[ MD(normalized ~K_n) = max_{s,t\in V(K_n)}E \big(\delta (s,t) \big) \leq \frac{n}{2} \]

Since this is only an upper bound, we wonder if we can find temporal paths with smaller arrival time than that bound. Indeed, we give a \textit{simple (greedy)} algorithm which can, with high probability, find a journey with small expected arrival time from a given source vertex $s$ to a given target vertex $t$ in the normalized uniform random temporal clique.

\noindent \textit{Note.} From here on, the notation ``$\log$'' will denote the natural logarithm.

\newpage
\begin{algorithm}[h]
\caption{The normalized U-RTG clique short journey finding algorithm, Extend-Try}
\label{alg:extend}
\begin{algorithmic}[1]
\Procedure {Extend-Try}{$clique~ K_n$, $s$, $t$, $c_1$, $k$}
	\For {i = 0 ... $c_1 \sqrt{n} \log{n}$}
		\State $s_i$ := undefined;
	\EndFor
	\State $s_0$ := $s$;
	\For {i = 0 ... $c_1 \sqrt{n} \log{n}$} \label{lst:line:for1}
		\If {$l(\{s_i,t\}) \in \big(c_1 \sqrt{n} (\log{n}) k, c_1 \sqrt{n} (\log{n}) k + \sqrt{n}\big)$}
			\State Follow directly the edge $\{s_i,t\}$; \textcolor{blue!70}{Success!}
			\State go to line \ref{lst:line:end1}
		\Else
			\If {$\exists u \in U\setminus \{t\}$ (where U stands for the set of the unvisited $~~~~~~~~~~~~~$ vertices) such, that $l(\{s_i,u\}) \in \big(k\cdot i , k (i+1) \big)$}
				\State $s_{i+1} = u$;
				\State go to line \ref{lst:line:for1}
			\Else
				\State follow directly the edge $\{s_i, u\}$ with the smallest $l(\{s_i, u\})$ $~~~~~~~~~~~~~~~~~~$ among all $u\in U$; \textcolor{blue!70}{Failure!}
				\State go to line \ref{lst:line:end2}
			\EndIf
		\EndIf	
	\EndFor
\For {i = 0 ... $c_1 \sqrt{n} \log{n}$}\label{lst:line:end1}
	\State \textbf{return} $s_i$;
\EndFor
\EndProcedure \label{lst:line:end2}
\end{algorithmic}
\end{algorithm}

\paragraph{Analysis of Extend-Try}

Next, we analyze algorithm \ref{alg:extend}, looking for the probability that it succeeds.

The probability that the time label of the edge $\{s_i,t\}$ belongs to the interval $(c_1 \sqrt{n} k, c_1 \sqrt{n} k + \sqrt{n})$ and thus the algorithm succeeds in the  $(i+1)^{\text{th}}$ iteration, is:
\[ P\Big(  l(\{s_i,t\}) \in \big(c_1 \sqrt{n} (\log{n}) k, c_1 \sqrt{n} (\log{n}) k + \sqrt{n}\big)  \Big) = \frac{\sqrt{n}}{n} = \frac{1}{\sqrt{n}} \]

\noindent Let $\varepsilon_1^j$ be the following event:
\[\text{``The algorithm finds a proper journey $s_0s_1,s_1s_2,s_2,s_3, \ldots, s_{j-1}s_j$''}\]
meaning that it finds a temporal path, on the temporal edges of which we find strictly ascending time labels and in fact the $i^{th}$ temporal edge's time label correctly belongs to the interval $((i-1) k, i k)$. The time labels are given to the edges independently from one another, thus the probability that the event $\varepsilon_1^j$ occurs is the product of the following probabilities:
$P\Big(\exists s_1\text{ unvisited vertex } :~\text{the edge } \{s_0,s_1\} \text{ has time label }l(\{s_0,s_1\}) \in (0,k)  \Big)$\\
$P\Big(\exists s_2\text{ unvisited vertex } :~\text{the edge } \{s_1,s_2\} \text{ has time label }l(\{s_1,s_2\}) \in (k,2k) $\\
$\vdots$\\
$P\Big(\exists s_j \text{ unvisited vertex }:~\text{the edge } \{s_{j-1},s_j\} \text{ has time label }l(\{s_{j-1},s_j\}) \in \big( (j-1)k,jk \big)$\\[1cm]
For any $i^{th}$ probability of the above, it holds that:

\begin{IEEEeqnarray*}{Cl}
 & 	P\Big(\exists s_i \text{ unvisited vertex}  :~\text{the edge }  \{s_{i-1},s_i\} \text{ has } l( \{s_{i-1},s_i\} ) \in ((i-1) k , i k)  \Big)
\\
= & 	1 - P\Big( \not\exists s_i \text{ unvisited vertex}:~\text{the edge } \{s_{i-1},s_i\} \text{ has }l(\{s_{i-1},s_i\}) \in ((i-1) k , i k)  \Big)	
\\
= & 	1 - P\Big(\forall s_i \text{ unvisited vertices}:~\text{the edge } \{s_{i-1},s_i\} \text{ has }l(\{s_{i-1},s_i\}) \notin ((i-1) k , i k)  \Big)	
\\
= & 	1 - \Big(  P\big(    \text{the edge } \{s_{i-1},s_i\} \text{ has }l(\{s_{i-1},s_i\}) \notin ((i-1) k , i k) 		,s_i \text{ unvisited vertex}			 \big)  \Big)^{n-i}	
\\
= & 	1 - \Big(1-  P\big(    \text{the edge } \{s_{i-1},s_i\} \text{ has }l(\{s_{i-1},s_i\}) \in ((i-1) k , i k) 		,s_i \text{ unvisited vertex}			 \big)  \Big)^{n-i}	
\\
= & 1- \Big( 1 - \frac{k}{n} \Big) ^{n-i}
\end{IEEEeqnarray*}

\noindent Therefore, the probability that $\varepsilon_1^j$ occurs, is:
\begin{IEEEeqnarray*}{lCl}
P(\varepsilon_1^j)  & = & \Bigg(	1- \Big( 1- \frac{k}{n} \Big)^{n-1}	\Bigg) \cdot
\\
&& \:  \Bigg(	1- \Big( 1- \frac{k}{n} \Big)^{n-2}	\Bigg)	\cdot \ldots \cdot 		\Bigg(	1- \Big( 1- \frac{k}{n} \Big)^{n-j}	\Bigg) \geq
\\
& \geq & \Bigg(	1- \Big( 1- \frac{k}{n} \Big)^{n-j}	\Bigg)^j \geq
\\
& \geq & \Bigg(		1-e^{-k}\Big(1-\frac{k}{n}\Big)^{-j} 		\Bigg)^j 	
\end{IEEEeqnarray*}

For $j \leq c_1 \sqrt{n} \log{n}$, we have:
\begin{IEEEeqnarray*}{rrCll}
 & \Big(1- \frac{k}{n}\Big)^{-j} & \leq & \Big(1-\frac{k}{n}\Big)^{-c_1 \sqrt{n} \log{n}} & \Leftrightarrow
\\
\Leftrightarrow & 1-e^{-k}\Big(1-\frac{k}{n}\Big)^{-j} & \geq & 1-e^{-k} \Big(1-\frac{k}{n}\Big)^{-c_1\sqrt{n}\log{n}} & 
\end{IEEEeqnarray*}

and:
\begin{IEEEeqnarray*}{rCl}
\Bigg(		1-e^{-k}\Big(1-\frac{k}{n}\Big)^{-j} 		\Bigg)^j 		& \geq &		 \Bigg(		1-e^{-k} \Big(1-\frac{k}{n}\Big)^{-c_1\sqrt{n}\log{n}}	\Bigg)^{c_1\sqrt{n}\log{n}}
\end{IEEEeqnarray*}

As a result, for $j \leq c_1 \sqrt{n} \log{n}$, it is:
\begin{IEEEeqnarray*}{lCl}
P(\varepsilon_1^j)  & \geq  &  	\Bigg(		1-e^{-k} \Big(1-\frac{k}{n}\Big)^{-c_1\sqrt{n}\log{n}}	\Bigg)^{c_1\sqrt{n}\log{n}}\\
P(\varepsilon_1^j)  & \geq  &  	1-e^{-k} \Big(1-\frac{k}{n}\Big)^{-c_1\sqrt{n}\log{n}}
\end{IEEEeqnarray*}

It holds asymptotically:
\begin{IEEEeqnarray*}{rCll}
c_1\sqrt{n}\log{n}  & \leq  &  n 	& \Leftrightarrow
\\
\Big(1-\frac{k}{n}\Big)^{c_1\sqrt{n}\log{n}}		&	\geq   &			\Big(1-\frac{k}{n}\Big)^{n}	& \Leftrightarrow
\\
\Big(1-\frac{k}{n}\Big)^{-c_1\sqrt{n}\log{n}}		&	\leq   &			\Big(1-\frac{k}{n}\Big)^{-n}	& \Leftrightarrow
\\
1-e^{-k} \Big(1-\frac{k}{n}\Big)^{-c_1\sqrt{n}\log{n}}		&	\geq		&	1-e^{-k}\Big(1-\frac{k}{n}\Big)^{-n}
\end{IEEEeqnarray*}

Therefore:
\begin{IEEEeqnarray*}{lCl}
P(\varepsilon_1^j)  & \geq  &  	1-e^{-k}\Big(1-\frac{k}{n}\Big)^{-n}
\end{IEEEeqnarray*}

and since $k \geq 1$, we have:
\begin{IEEEeqnarray*}{lCl}
P(\varepsilon_1^j)  & \geq & 					1-e^{-k}\Big(1-\frac{1}{n}\Big)^{-n}
\\
 & = & 					1-e^{-k}e
\\
 & = & 					1-e^{1-k}
\end{IEEEeqnarray*}

For $k=r\log{n},~r>1$, we have:
\begin{IEEEeqnarray*}{lCl}
P(\varepsilon_1^j)  & \geq  &  1-e^{1-r\log{n}}
\\
 & = & 1-en^{-r}
\end{IEEEeqnarray*}

The probability that we fail in every iteration $i=0, \ldots, c_1 \sqrt{n}\log{n}$ to find a vertex $s_i$ such, that $l(\{s_i,t\}) \in (c_1 \sqrt{n} k, c_1 \sqrt{n} k + \sqrt{n})$ is:
\begin{IEEEeqnarray*}{lCl}
P(all fail)  & =  &  \overbrace{
\Big( 1 - \frac{1}{\sqrt{n}}  \Big) \cdot	\Big( 1 - \frac{1}{\sqrt{n}}  \Big) \cdot	\ldots  \Big( 1 - \frac{1}{\sqrt{n}}  \Big)}^{c_1 \sqrt{n}\log{n} \text{ factors}}
\\
 & = & \Big( 1 - \frac{1}{\sqrt{n}}  \Big)^{c_1 \sqrt{n}\log{n}}
\\
 & = &  e^{-c_1 \log{n}} = n^{-c_1}
\end{IEEEeqnarray*}

The probability that we succeed in some iteration of the algorithm is:
\begin{IEEEeqnarray*}{lCl}
P(success) & = & \Big( 1- P(all fail) \Big) P(\varepsilon_1^j)  
\\
 & \geq &	\Big(	1-n^{-c_1}	\Big) \Big(	1-en^{-r}	\Big)
\end{IEEEeqnarray*}

Therefore, the following theorem holds:\\[0.5cm]

\begin{thm} \label{1}
For any constants $c_1, r > 1$, given two vertices $s, t$, $s \not= t$, of the normalized uniform random temporal clique, $K_n$, the probability to arrive, starting from $s$, to $t$ at time at most \[t_0 = c_1 \sqrt{n} (\log{n}) k  + \sqrt{n} \text{, where }k=r\log{n}\] is at least \[ \Big(	1-n^{-c_1}	\Big) \Big(	1-en^{-r}	\Big).\]
\end{thm}

$~$\\

\noindent \textit{Remark.} 
For the ``on-line" case, where a traveller starts from $s$ and wants to find a small journey to $t$, but he can only see the edges (\textit{arcs}) out of visited vertices, we conjecture that Algorithm \ref{alg:extend} gives a very tight bound on the expected arrival time.

\section{Temporal Diameter}\label{sec:td}
In this section, we study the concept of the temporal diameter of a uniform random temporal graph.

\begin{mydef}
Consider an instance $G(L)$ of a U-RTG. We denote the maximum of all distributional temporal distances between all pairs of vertices of $G(L)$ by $d(G(L))$:
\[ d(G(L)) = max_{s,t \in V(G)} \delta ' (s,t). \]
We define $diam(G(L)) = min\{d(G(L)), n'\}$.
Then, the Expected or Temporal Diameter of G, denoted by $TD$, is given by the following formula:
\[TD(G) = E\Big(diam\big(G(L)\big)\Big) =\sum_L diam\big(G(L)\big) \cdot P(L)\]
, where $P(L)$ is the probability for labelling $L$ to occur.
\end{mydef}

We can easily prove that every temporal graph's temporal diameter, $TD$, is equal or greater than its maximum expected temporal distance, $MD$. 

\begin{thm} It holds that:
\[ TD(G) \geq MD(G),\text{ for every temporal graph } G. \]
\end{thm}

\begin{proof}
To prove this, we use the Reverse Fatou's Lemma\cite{durrett}:\\

\begin{them} [Reverse Fatou's lemma]
 If $X_n\geq 0$, for all $n$, then
\[   E(lim_n sup X_n) \geq lim_n sup E(X_n).  \]
\end{them}

In other words, the expected value of the maximum of a set of random variables is at least equal to the maximum of the expected values of those variables.

Now, notice that the Temporal Diameter of a temporal graph $G$ is actually the expected value of the maximum of all distributional temporal distances, that is $E(max_{s,t \in V(G)} \delta ' (s,t))$, in the case where we have $\delta ' (s,t) \leq n'$, for every pair of vertices $s,t \in V(G)$. In that case, the Maximum Expected Temporal Distance of $G$ is actually the maximum of the expected values of all pairs of vertices' distributional temporal distances, that is $max_{s,t \in V(G)} E(\delta ' (s,t))$. Therefore, in that case, using the above described Reverse Fatou's Lemma, we conclude that:
\[ TD(G) \geq MD(G). \]

In the case, where there is at least one pair of vertices $s,t \in V(G)$ such, that $\delta '(s,t) \geq n'$, both the temporal diameter and the maximum expected temporal distance of $G$ are equal to $n'$.

Thus, we conclude that it generally applies that:
\[ TD(G) \geq MD(G),\text{ for every temporal graph } G. \]
\end{proof}

We will now prove that the time $t_0 - o(t_0)$ (see. Theorem \ref{1}) is an upper bound of the normalized uniform random temporal clique's temporal diameter, $TD$, and, thus, is an upper bound of its maximum expected temporal distance, $MD$. 

\begin{thm}
The quantity $t_0 - o(t_0)$ is an upper bound of both the temporal diameter, $TD$, and the maximum expected temporal distance, $MD$, of the normalized U-RT clique.
\end{thm}

\begin{proof}
Let $s,t$ be two vertices of the normalized U-RT clique. We call $E_{st}$ the following event:
\[ \text{``We arrive, starting from $s$, to $t$ at time at most $t_o$''}\]
where $t_0 = c_1 \sqrt{n} (\log{n}) k + \sqrt{n}, ~ c_1 >1, k= r\log{n},~ r>1$.\\
It holds that:

\begin{IEEEeqnarray*}{lCl}
P(E_{st}) & \geq &  \Big(	1-n^{-c_1}	\Big) \Big(	1-en^{-r}	\Big) 
\\
 & \geq &	1- n^{-c_1} - e n ^{-r}
\end{IEEEeqnarray*}

For $r=c_1$, the above relation becomes:

\begin{IEEEeqnarray*}{lCl}
P(E_{st})  & \geq &	1- n^{-c_1} - e n ^{-c_1}
\\
 & \geq & 1- 2 e n ^{-c_1}
\end{IEEEeqnarray*}

Therefore, the probability that the complement of $E_{st}$ occurs is:
\begin{IEEEeqnarray*}{lCl}
P(\overline{E}_{st})  &= &	1- P(E_{st}) 
\\
 & \leq & 2 e n ^{-c_1}
\end{IEEEeqnarray*}

Thus, the probability that there exist two vertices $s,t$ such that we arrive, starting from $s$, to $t$ at time greater than $t_0$ is:

\begin{IEEEeqnarray*}{lCl}
P(\exists s,t : \overline{E}_{st})  & \leq &	n (n-1) 2 e n ^{-c_1}
\\
 & \leq & 2 e n ^{-c_1 -2}
\end{IEEEeqnarray*}

Let us denote by $T$ the $max\{ a_{st}, s,t \in V(K_n)\}$, where $a_{st}$ is the greatest arrival time amongst all ($s,t$)-journeys' arrival times. Then, we have:

\begin{IEEEeqnarray*}{lCl}
P(\exists s,t : \overline{E}_{st})  & = &	P(T > t_0)
\\
 & \leq & 2 e n ^{-c_1 -2}
\end{IEEEeqnarray*}

It is:
\begin{IEEEeqnarray*}{lCl}
TD  & \leq &	E(max\{ a_{st}, s,t \in V(K_n)\})
\\
 & \leq & ( 1 - 2 e n ^{-c_1 -2} ) \cdot t_0 + n \cdot 2 e n ^{-c_1 -2}
\\
& \leq & t_0 -o(t_0)
\end{IEEEeqnarray*}

Since $TD(G) \geq MD(G),\text{ for every temporal graph } G $, we conclude that in the case of the normalized U-RT clique, it is:

\[  MD \leq TD \leq t_0 -o(t_0) \]

\end{proof}

\section{An optimization problem: The Bridges' problem} \label{sec:bridge}
We will now study an optimization problem concerning the temporal multigraph shown in Figure \ref{fig:poly1}.
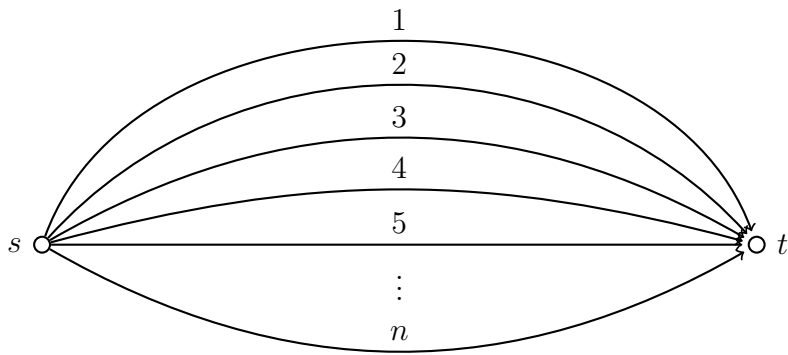
\begin{figure}[htbp]
\begin{center}

\[\begin{tikzpicture}[thick,scale=0.95,->,shorten >=2pt]
	\vertex (1) at (0,0) [label=left:$s$] {};
	\vertex (2) at (10,0) [label=right:$t$] {};
	\path
		(1) edge [bend left=70] node[above]{$1$}  (2)
		(1) edge [bend left=48] node[above]{$2$} (2)
		(1) edge [bend left=30] node[above]{$3$} (2)
		(1) edge [bend left=15] node[above]{$4$} (2)
		(1) edge node[above]{$5$} (2)
		(1) edge [line width=0pt,black!0.05,bend right=18] node[above]{\textcolor{black}{$\vdots$}} (2)
		(1) edge [bend right=30] node[above]{$n$} (2)
	;
\end{tikzpicture}\]
\end{center}

\rule{35em}{0.5pt}
\caption{The bridges' problem}
\label{fig:poly1}
\end{figure}

\noindent\underline{\textbf{\textit{The problem}}}\\
$n$ people are located on one bank of a river (see vertex $s$, Figure \ref{fig:poly1}) and want to go to the other side (see vertex $t$, Figure \ref{fig:poly1}). Each one can go across one of a total of $n$ bridges that connect the two riversides, paying individual cost equal to $\displaystyle{1 + \frac{i}{m_i}}$, where $i$ stands for the number of the bridge they pass and $m_i$ stands for the total sum of people that cross that bridge. Thus, the total cost payed by $m$ people to cross the \textbf{$i^{th}$} bridge, $i=1,2,\ldots, n$, is:\[cost[i]=m_i+i\]

\noindent We denote by \emph{maximum cost payed} the maximum, over all bridges $i$, cost $m_i +i$:
\[maximum~cost~payed = max\{m_i + i, i=1,2, \ldots, n :~bridge\} \]

\noindent How should the $n$ people be assigned to the bridges so, that the maximum cost payed is minimized?\\
We denote the minimum, over all assignments of $n$ people to $n$ bridges, maximum cost payed by $OPT$, that is:
\[ OPT = min_{all~ assignments}\{ maximum ~cost~ payed\} \]

\noindent \textit{Remark.} In another interpretation of the bridges' problem, as we call the above described problem, we consider the multi-labeled temporal digraph of two vertices $s,t$ and one single edge $\{s,t\}$ which is assigned the discrete time labels $1,2,\ldots, n$ (see figure \ref{fig:poly2}).

\begin{figure}[htbp]
\begin{center}

\[\begin{tikzpicture}[thick,scale=0.95,->,shorten >=2pt]
	\vertex (1) at (0,0) [label=left:$s$] {};
	\vertex (2) at (6,0) [label=right:$t$] {};
	\path
		(1) edge node[above]{$1, 2, \ldots, n$}  (2)
	;
\end{tikzpicture}\]
\end{center}

\rule{35em}{0.5pt}
\caption{The bridges' problem (another interpretation)}
\label{fig:poly2}
\end{figure}
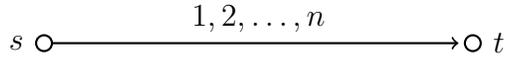

\noindent Here we have a single bridge which is available everyday from day $1$ to day $n$. As time progresses the cost someone needs to pay to move from $s$ to $t$ increases. Again, one has to pay individual cost equal to $\displaystyle{1 + \frac{i}{m_i}}$, where $i$ stands for the day on which he decides to move from $s$ to $t$ and $m_i$ stands for the total sum of people decide to move from $s$ to$t$ on that same day. Therefore, the total cost payed by $m$ people who move from $s$ to $t$ on the \textbf{$i^{th}$} day, $i=1,2,\ldots, n$, is:\[cost[i]=m_i+i\]

\begin{thm}
We can compute the assignment of $n$ persons to $n$ bridges that achieves the $OPT$ in polynomial time $O(n^2)$.
\end{thm}
\begin{proof}
We provide Algorithm \ref{alg:bridge} and show that it computes the assignment that achieves $OPT$.

\begin{algorithm}[h]
\caption{The bridges problem solving algorithm}
\label{alg:bridge}
\begin{algorithmic}[1]
\Procedure {Bridges}{$n$}
	\State cost[] is a 1$\times$n array which holds the bridges' costs;
	\State content[] is a 1$\times$n array which holds the bridges' contents;		 \Comment{a.k.a \\ \hfill  how\\ \hfill many people\\ \hfill  are on each\\ \hfill  bridge}
	\State m := n;														 \Comment{m is the number of bridges}
	\For	{i = 1 ... m}
		\State content[i] := 0;											 \Comment{Initializations}
		\State cost[i] := i;
	\EndFor

\algstore{alg:bridge}
  \end{algorithmic}
\end{algorithm}

\clearpage

\begin{algorithm}
  \ContinuedFloat
  \caption{The bridges problem solving algorithm (continued)}
  \begin{algorithmic}
      \algrestore{alg:bridge}

	\For {i = 1 ... n}
		\State bridge := 1; 												 \Comment{Initialize the bridge that the $i^{th}$ person will pass}
		\For { j = 2 ... m}												 \Comment{Find the bridge that gives the minimum\\ \hfill  possible cost}
			\If { cost[j] $<$ cost[bridge] }
				\State bridge := j;
			\EndIf
		\EndFor
		\State content[bridge] := content[bridge]+1;						 \Comment{Add the $i^{th}$ person to \\ \hfill the selected \\ \hfill bridge's content}
		\State cost[bridge] := cost[bridge]+1;								 \Comment{Calculate the right new cost}
	\EndFor

	\For {i = 1 ... m}
		\If {content[i] == 0}
			\State cost[i] := 0;
		\EndIf
		\If {content[i] == 1}
			\State \textbf{Write} content[i] , `` person passes bridge \#'', i , `` who $~~~~~~~$ $~~~~~~~~~~~~~~$ therefore has to pay cost equal to '', cost[i];

		\Else
			\State \textbf{Write} content[i] , `` people pass bridge \#'', i , `` who therefore  $~~~~~~~~~~~~~~~~$ have to pay cost equal to '', cost[i];		
		\EndIf			 \Comment{Print the bridges' costs}

	\EndFor
\EndProcedure
\end{algorithmic}
\end{algorithm}

The algorithm assigns the $i^{th}$ person to the bridge, for which the current minimum cost is payed. If there are more than one such bridges, the algorithm assigns the $i^{th}$ person to the first one in order. It is trivial to see that the algorithm's running time is $O(n^2)$.\\

\noindent\textbf{Proof of correctness} 
We will prove the validity of the algorithm \ref{alg:bridge} by induction on the number $n$ of persons.

\begin{itemize}
\item For $n=1$, the algorithm sets the number of bridges to be $m=1$ and the sole bridge's content and cost to be equal to 1. In the main loop, the sole person is assigned to the bridge, paying cost equal to:\[cost[1] = 2\]
So, actually, the algorithm solves the problem for $n=1$ person.
\item Assume that the algorithm solves the problem for $n=k$ people.
\item We will show that the algorithm solves the problem for $n=k+1$ people. 

Before continuing, let us consider the following: Let $n_1, n_2 \in \mathbb{N}$ numbers of people, with $n_1>n_2$. It is obvious that the minimum possible maximum cost for $n=n_1$ people is at least equal to the minimum possible maximum cost for $n=n_2$ people.

Let us observe now that the procedures performed by the algorithm in the main loop for $k$ people, and the results obtained through these, are identical to those performed and obtained respectively for $k+1$ people, except that for $k+1$ people, there is a $(k+1)^{th}$ bridge, which throughout the execution of these processes has zero content, and there is also an additional execution of the loop. At the beginning of this $(k+1)^{th}$ execution, the algorithm has already assigned the $k$ people to the brisges in a way that we obtain the minimum possible maximum cost.

The algorithm, by construction, assigns the people to the bridges in a way that their costs are ordered by (not necessarily strictly) descending order and indeed one of the following two possible events occur:
\begin{equation*}
\left\{
\begin{array}{l}
\text{all the bridges have the same cost, denoted by } OPT\\
\text{or}\\
\text{some bridges have cost } OPT \text{ and some others have cost }OPT-1.
\end{array} \right.
\end{equation*}
In the second case, the algorithm is obviously going to assign the $(k+1)^{th}$ person to the first in order bridge that has cost equal to $OPT-1$, thereby maintaining the maximum cost that occurs on the bridges to a minimum, that is $OPT$.

In the first case, if $r \leq k+1$ is the number of the last bridge that has positive content,  $content[r]$, then it is:

\begin{equation*}
\left\{
\begin{array}{l}
r+content[r] = OPT\\
\text{But: } content[r] \geq 1 \text{ and so: } r+ content[r] \geq  r+1
\end{array} \right\}
\Rightarrow OPT \geq r+1
\end{equation*}

Also, since the$(r+1)^{th}$ bridge has zero content, it is:\[cost[r+1] = r+1\]
The algorithm checks which of the $ k +1 $ bridges has the minimum cost to assign the $(k+1)^{th}$ person to that bridge. If $OPT=r+1$, then the algorithm assigns the last person to the $1^{st}$ bridge. Otherwise, it assigns it to the $ (r +1)^{th} $ bridge. This way, it ensures the minimum possible maximum cost for the $ k +1 $ bridges.

Therefore, the algorithm solves the problem for $ n = k +1 $ people.
\end{itemize}

\end{proof}

We will now calculate the value of the $OPT$. Again, let us denote by $r$ the number of bridges that have a positive content, i.e. are not empty, in the optimal case which the Algorithm \ref{alg:bridge} computes. For the sake of brevity, let us also denote by $l_i$ the content of the $i^{th}$ bridge. Since the average cost of the non empty bridges is equal or less than the maximum cost that occurs on those bridges, the following holds for the optimal case:
\[ \frac{\displaystyle\sum_{i=1}^{r} (i+l_i)}{r} \leq OPT \]
Therefore, we have:
\begin{equation}\label{eq:4}
\displaystyle\sum_{i=1}^{r} (i+l_i) \leq r OPT
\end{equation}

Furthermore, it is easy to see that since, in the optimal case that the algorithm computes, the $OPT$ is greater than any bridge's cost by at most \textit{one}, it holds that:
\begin{equation}\label{eq:5}
 rOPT - r \leq \displaystyle\sum_{i=1}^{r} (i+l_i) 
\end{equation}

By the relations \eqref{eq:4} and \eqref{eq:5}, we have:
\begin{IEEEeqnarray*}{lCCCl l}
rOPT - r & \leq & \displaystyle\sum_{i=1}^{r} (i+l_i)  & \leq & rOPT & \Leftrightarrow
\\
rOPT - r & \leq & \displaystyle\sum_{i=1}^{r} i +\displaystyle\sum_{i=1}^{r} l_i  & \leq & rOPT & \Leftrightarrow
\\
rOPT - r & \leq & \frac{r(r+1)}{2} +n & \leq & rOPT & \Leftrightarrow
\\
OPT - 1 & \leq & \frac{(r+1)}{2} +\frac{n}{r} & \leq & OPT 
\end{IEEEeqnarray*}

Now, the quantity $ \frac{(r+1)}{2} +\frac{n}{r} $ is minimized at $r=\sqrt{2n}$ and at that point, its value is equal to $\sqrt{2n} +\frac{1}{2}$. Therefore, we conclude that:
\[OPT = \lceil \sqrt{2n} +\frac{1}{2} \rceil \]

\section{Conclusions and further research}\label{sec:concl}
There are several open problems related to the findings of the present work. We initiated here the random availability of edges where the selection of time-labels, and thus the selection of moments in time at which the edges are available, follows the uniform distribution. There are still other interesting approaches concerning what distribution the selection of time-labels could follow (see F-CASE in Section \ref{sec:def}). Another approach that is yet to be examined is that of the multi-labeled temporal graphs, on which we could search for statistical properties respective to  the ones we studied within the present work. Yet another interesting direction which we did not consider in this work is to find upper bounds on the maximum expected temporal distance and the temporal diameter of any U-RTG (or F-RTG). Further research could also focus on calculating the actual value of these properties, e.g. in the case of the normalized uniform random temporal clique.

\end{document}